\documentclass[11pt]{article} \usepackage[normal]{optional}
\usepackage{amsmath}
\usepackage{amssymb}
\usepackage{ifpdf}
\usepackage{cite}
\usepackage{pageno}
\usepackage{comment}
\usepackage{import}
\usepackage{caption}
\usepackage{subcaption}

\usepackage[text={6.5in,9in},centering]{geometry}
\opt{normal}{\usepackage{commands-tam,numbering-tam}}
\opt{submission}{
    \usepackage{commands-tam}
    \numberwithin{equation}{section}
    \numberwithin{theorem}{section}
    \spnewtheorem{cor}[theorem]{Corollary}{\bfseries}{\itshape}
    \spnewtheorem{lem}[theorem]{Lemma}{\bfseries}{\itshape}
    \spnewtheorem{prop}[theorem]{Proposition}{\bfseries}{\itshape}
    \spnewtheorem{obs}[theorem]{Observation}{\bfseries}{\itshape}
}

\renewcommand{\vec}[1]{\mathbf{#1}}

\newcommand\addtag{\refstepcounter{equation}\tag{\theequation}}

\usepackage{url}
\usepackage{amsfonts}
\opt{normal}{\usepackage{amsthm}}
\opt{normal}{\usepackage{fullpage}}

\opt{submission,final}{}
\opt{normal}{}

\renewcommand{\exp}[1]{\mathrm{E}\left[ #1 \right]}

\newcommand{\ignore}[1]{}

\usepackage[normalem]{ulem}

\newcommand{\calE}{\mathcal{E}}

\newcommand{\hf}{\hat{f}}
\newcommand{\hr}{\hat{r}}
\newcommand{\hatt}{\hat{t}}
\newcommand{\bT}{\mathbf{T}}

\newcommand{\bD}{\mathbf{D}}
\newcommand{\bW}{\mathbf{W}}
\newcommand{\bU}{\mathbf{U}}

\newcommand{\bF}{\mathbf{F}}
\newcommand{\bR}{\mathbf{R}}
\newcommand{\bP}{\mathbf{P}}
\newcommand{\bX}{\mathbf{X}}

\newcommand{\vc}{\vec{c}}
\newcommand{\vr}{\vec{r}}

\newcommand{\vp}{\vec{p}}
\newcommand{\vi}{\vec{i}}

\renewcommand{\vb}{\vec{b}}
\newcommand{\bfr}{\mathbf{r}}
\newcommand{\bfp}{\mathbf{p}}
\renewcommand{\Pr}{\mathsf{Pr}}

\renewcommand{\exp}{\mathrm{exp}}

\def\longrightharpoonup{\relbar\joinrel\rightharpoonup}
\def\longleftharpoondown{\leftharpoondown\joinrel\relbar}
\def\longrightleftharpoons{\mathop{\vcenter{\hbox{\ooalign{\raise1pt\hbox{$\longrightharpoonup\joinrel$}\crcr\lower1pt\hbox{$\longleftharpoondown\joinrel$}}}}}}

\ifpdf

  \usepackage[pdftex]{epsfig}
  \usepackage[pdfstartview={XYZ null null 1.00}, pdfpagemode=None,colorlinks=true,citecolor=blue,linkcolor=blue]{hyperref}

\else
    \usepackage[dvips]{epsfig}
\usepackage[dvips,colorlinks=true,citecolor=blue,linkcolor=blue]{hyperref}

\fi

\begin{document}

\title{Timing in chemical reaction networks
}

\opt{normal}{
    \author{
        David Doty\thanks{California Institute of Technology, Pasadena, CA, USA, {\tt ddoty@caltech.edu}. The author was supported by the Molecular Programming Project under NSF grant 0832824, a Computing Innovation Fellowship under NSF grant 1019343 and by NSF grants CCF-1219274 and CCF-1162589.}
    }
    \date{}
}

\maketitle

\begin{abstract}
  Chemical reaction networks (CRNs) formally model chemistry in a well-mixed solution.
  CRNs are widely used to describe information processing occurring in natural cellular regulatory networks, and with upcoming advances in synthetic biology, CRNs are a promising programming language for the design of artificial molecular control circuitry.
  Due to a formal equivalence between CRNs and a model of distributed computing known as \emph{population protocols}, results transfer readily between the two models.

  We show that if a CRN respects \emph{finite density} (at most $O(n)$ additional molecules can be produced from $n$ initial molecules), then starting from any \emph{dense} initial configuration (all molecular species initially present have initial count $\Omega(n)$, where $n$ is the initial molecular count and volume), then \emph{every} producible species is produced in constant time with high probability.

  This implies that no CRN obeying the stated constraints can function as a \emph{timer}, able to produce a molecule, but doing so only after a time that is an unbounded function of the input size.
  This has consequences regarding an open question of Angluin, Aspnes, and Eisenstat concerning the ability of population protocols to perform fast, reliable leader election and to simulate arbitrary algorithms from a uniform initial state.
\end{abstract}





\section{Introduction}
\label{sec-intro}

\subsection{Background of the field}

The engineering of complex artificial molecular systems will require a sophisticated understanding of how to \emph{program chemistry}.
A natural language for describing the interactions of molecular species in a well-mixed solution is that of (finite) chemical reaction networks (CRNs), i.e., finite sets of chemical reactions such as $A + B \to A + C$.
When the behavior of individual molecules is modeled, CRNs are assigned semantics through \emph{stochastic chemical kinetics}~\cite{Gillespie77}, in which reactions occur probabilistically with rate proportional to the product of the molecular count of their reactants and inversely proportional to the volume of the reaction vessel.
The kinetic model of CRNs is based on the physical assumption of well-mixedness valid in a dilute solution.
Thus, we assume the \emph{finite density constraint}, which stipulates that a volume required to execute a CRN must be proportional to the maximum molecular count obtained during execution~\cite{SolCooWinBru08}.
In other words, the total concentration (molecular count per volume) is bounded.
This realistically constrains the speed of the computation achievable by CRNs.

Traditionally CRNs have been used as a descriptive language to analyze naturally occurring chemical reactions, as well as numerous other systems with a large number of interacting components such as gene regulatory networks and animal populations. 
However, recent investigations have viewed CRNs as a programming language for engineering artificial systems.
These works have shown CRNs to have eclectic algorithmic abilities.
Researchers have investigated the power of CRNs to simulate Boolean circuits~\cite{magnasco1997chemical}, neural networks~\cite{hjelmfelt1991chemical}, and digital signal processing~\cite{riedelDSP2012}.
Other work has shown that bounded-space Turing machines can be simulated with an arbitrarily small, non-zero probability of error by a CRN with only a polynomial slowdown \cite{angluin2006fast}.\footnote{This is surprising since finite CRNs necessarily must represent binary data strings in a unary encoding, since they lack positional information to tell the difference between two molecules of the same species.}
Space- and energy-efficient simulation of space-bounded Turing machines can be done by CRNs implementable by logically and thermodynamically reversible DNA strand displacement reactions, and as a consequence it is $\PSPACE$-hard to predict whether a particular species is producible~\cite{ThaCon12}.
Even Turing-universal computation is possible with an arbitrarily small probability of error over all time~\cite{SolCooWinBru08}.
The computational power of CRNs also provides insight on why it can be computationally difficult to simulate them~\cite{soloveichik2009robust},
and why certain questions are frustratingly difficult to answer (e.g.\ undecidable)\cite{CooSolWinBru09,zavattaro2008termination}.
The programming approach to CRNs has also, in turn, resulted in novel insights regarding natural cellular regulatory networks~\cite{cardelli2012cell}.


\subsection{Focus of this paper}
Informally, our main theorem shows that
CRNs ``respecting finite density'' (meaning that the total molecular count obtainable is bounded by a constant times the initial molecular count)
starting from any \emph{dense} initial configuration (meaning that any species present initially has count $\Omega(n)$, where $n$ is the initial molecular count and volume) will create $\Omega(n)$ copies of \emph{every} producible species in constant time with high probability.
We now explain the significance of this result in the context of an open question regarding the ability of stochastic CRNs to do fast, reliable computation.

Many of the fundamental algorithmic abilities and limitations of CRNs were elucidated under the guise of a related model of distributed computing known as \emph{population protocols}.
Population protocols were introduced by Angluin, Aspnes, Diamadi, Fischer, and Peralta~\cite{angluin2006passivelymobile} as a model of resource-limited mobile sensors; see Aspnes and Ruppert~\cite{aspnes2007introduction} for an excellent survey of the model.
A population protocol consists of a set of $n$ agents with finite state set $\Lambda$ (where we imagine $n \gg |\Lambda|$), together with a transition function
$\delta:\Lambda^2\to\Lambda^2$,
with $(r_1,r_2)=\delta(q_1,q_2)$ indicating that if two agents in states $q_1$ and $q_2$ interact, then they change to states $r_1$ and $r_2$, respectively.
Multiple semantic models may be overlaid on this syntactic definition, but the randomized model studied by Angluin, Aspnes, and Eisenstat~\cite{angluin2006fast}, in which the next two agents to interact are selected uniformly at random, coincides precisely with the model of stochastic chemical kinetics, so long as the CRN's reactions each have two reactants and two products,\footnote{This turns out not to be a significant restriction on CRNs; see the introduction of~\cite{CheDotSol12} for a discussion of the issue.} i.e., $q_1 + q_2 \to r_1 + r_2$.
In fact, every population protocol, when interpreted as a CRN, respects finite density, because the total molecular count is constant over time.
Therefore, although we state our results in the language of CRNs, the results apply \emph{a fortiori} to population protocols.

Angluin, Aspnes, and Eisenstat~\cite{angluin2006fast}, and independently Soloveichik, Cook, Winfree, and Bruck~\cite{SolCooWinBru08} showed that arbitrary Turing machines may be simulated by randomized population protocols/CRNs with only a polynomial-time slowdown.
Because a binary input $x \in \{0,1\}^*$ to a Turing machine must be specified in a CRN by a unary molecular count $n \in \N$, where $n \approx 2^{|x|}$, this means that each construction is able to simulate each step of the Turing machine in time $O(\mathrm{polylog}(n))$, i.e. polynomial in $|x|$.
Both constructions make essential use of the ability to specify \emph{arbitrary} initial configurations, where a configuration is a vector $\vc \in \N^\Lambda$ specifying the initial count $\vc(S)$ of each species $S \in \Lambda$.
In particular, both constructions require a certain species to be present with initial count 1, a so-called ``leader'' (hence the title of the former paper).

A test tube with a single copy of a certain type of molecule is experimentally difficult to prepare.
The major open question of~\cite{angluin2006fast} asks whether this restriction may be removed: whether there is a CRN starting from a \emph{uniform} initial configuration (i.e., with only a single species present, whose count is equal to the volume $n$) that can simulate a Turing machine.

Angluin, Aspnes, and Eisenstat~\cite{angluin2006fast} observe that a leader may be elected from a uniform initial configuration of $n$ copies of $L$ by the reaction $L + L \to L + N$; in other words, when two candidate leaders encounter each other, one drops out.
However, this scheme has problems:
\begin{enumerate}
  \item The expected time until a single leader remains is $\Theta(n)$, i.e., exponential in $m$, where $m \approx \log n$ is the number of bits required to represent the input $n$.
      In contrast, if a leader is assumed present from the start, the computation takes time $O(t(\log n) \cdot \log^5 n)$, where $t(m)$ is the running time of the Turing machine on an $m$-bit input, i.e., the time is polynomial in $t$.
      Therefore for polynomial-time computations, electing a leader in this manner incurs an exponential slowdown.

  \item There is no way for the leader to ``know'' when it has been elected. Therefore, even if a single leader is eventually elected, and even if the CRN works properly under the assumption of a single initial leader, prior to the conclusion of the election, multiple leaders will be competing and introducing unwanted behavior.
\end{enumerate}

These problems motivate the need for a \emph{timer} CRN, a CRN that, from a uniform initial configuration of size $n$, is able to produce some species $S$, but the time before the first copy of $S$ is produced is $t(n)$, for some unbounded function $t:\N\to\N$.
If a CRN exists that can produce a copy of $S$ after $\Omega(n)$ time, then by adding the reaction $S + L \to L_\mathrm{active}$, a leader election can take place among ``inactive'' molecules of type $L$, and the remaining species of the leader-driven CRNs of~\cite{angluin2006fast} or~\cite{SolCooWinBru08} can be made to work only with the ``active'' species $L_\mathrm{active}$, which will not appear (with high probability) until there is only a single copy of $L$ to be activated.
This shows how a timer CRN could alleviate the second problem mentioned above.
In fact, even the first problem, the slowness of the na\"{i}ve leader election algorithm, is obviated in the presence of a timer CRN.
If a timer CRN produces $S$ after $t(n) = \Omega(\log n)$ time from a uniform initial configuration of $n$ copies of $X$, then this can be used to construct a CRN that elects a leader in time $O(t(n))$~\cite{ConDotTha13}.
In other words, the problem of constructing a timer CRN is ``leader-election-hard''.

Unfortunately, timer CRNs cannot be constructed under realistic conditions.
The main theorem of this paper, stated informally, shows the following (the formal statement is Theorem~\ref{thm-no-timer} in Section~\ref{sec-timer}):
Let $\calC$ be a CRN with volume and initial molecular count $n$ that \emph{respects finite density} (no sequence of reactions can produce more than $O(n)$ copies of any species) with an initial configuration that is \emph{dense} (all species initially present have initial count $\Omega(n)$).
Then with probability $\geq 1 - 2^{- \Omega(n)}$, every producible species is produced in time $O(1)$.
Since a uniform initial configuration is dense, timer CRNs that respect finite density cannot be constructed.
It should be noted that the condition of respecting finite density is physically realistic; it is obeyed, for instance, by every CRN that obeys the law of conservation of mass.
As noted previously, all population protocols respect finite density, which implies that no population protocol can function as a timer.

This theorem has the following consequences for CRNs respecting finite density that start from a dense initial configuration with $n$ molecules:
\begin{enumerate}
  \item
    No such CRN can decide a predicate ``monotonically''.
    By this we mean that if the CRN solves a decision problem (a yes/no question about an input encoded in its initial configuration) by producing species $Y$ if the answer is yes and species $N$ if the answer is no, then the CRN cannot be guaranteed to produce \emph{only} the correct species.
    It necessarily must produce \emph{both} but eventually dispose of the incorrect species.
    This implies that composing such a CRN with a downstream CRN that uses the answer necessarily requires the downstream CRN to account for the possibility that the upstream CRN will ``change its mind'' before converging on the correct answer.

  \item
    No CRN leader election algorithm can ``avoid war'': any CRN that elects a unique leader (in any amount of time) must necessarily have at least 2 (in fact, at least $\Omega(n)$) copies of the leader species for some time before the unique leader is elected.
    This implies that any downstream CRN requiring a leader must be designed to work in the face of multiple leaders being present simultaneously for some amount of time before a unique leader is finally elected.
\end{enumerate}

Another line of research that has illustrated a striking difference between ``chemical computation'' and standard computational models has been the role of randomness.
The model of stochastic chemical kinetics is one way to impose semantics on CRNs, but another reasonable alternative is to require that all possible sequences of reactions -- and not just those likely to occur by the model of chemical kinetics -- should allow the CRN to carry out a correct computation.
Under this more restrictive deterministic model of CRN computation, Angluin, Aspnes, and Eisenstat~\cite{AngluinAE2006semilinear} showed that precisely the \emph{semilinear predicates} $\phi:\N^k \to \{0,1\}$ (those definable in the first-order theory of Presburger arithmetic) are deterministically computed by CRNs.
This was subsequently extended to \emph{function} computation by Chen, Doty, and Soloveichik~\cite{CheDotSol12}, who showed that precisely the functions $f:\N^k\to\N^l$ whose graph is a semilinear set are deterministically computed by CRNs.
Since semilinear sets are computationally very weak (one characterization is that they are finite unions of ``periodic'' sets, suitably generalizing the definition of periodic to multi-dimensional spaces), this shows that introducing randomization adds massive computational ability to CRNs.

This strongly contrasts other computational models.
Finite automata decide the regular languages whether they are required to be deterministic or randomized (or even nondeterministic).
Turing machines decide the decidable languages whether they are required to be deterministic or randomized (or even nondeterministic).
In the case of polynomial-time Turing machines, it is widely conjectured~\cite{NisanWigderson94} that $\P=\BPP$, i.e., that randomization adds at most a polynomial speedup to any predicate computation.
Since it is known that $\P \subseteq \BPP \subseteq \EXP$, even the potential exponential gap between deterministic and randomized polynomial-time computation is dwarfed by the gap between deterministic CRN computation (semilinear sets, which are all decidable in linear time by a Turing machine) and randomized CRN computation (which can decide any decidable problem).

Along this line of thinking (``What power does the stochastic CRN model have over the deterministic model?''), our main theorem may also be considered a stochastic extension of work on deterministic CRNs by Condon, Hu, Ma\v{n}uch, and Thachuk~\cite{CondonHMT12} and Condon, Kirkpatrick, and Ma\v{n}uch~\cite{condon2012reachability}.
Both of those papers showed results of the following form: every CRN in a certain restricted class (the class being different in each paper, but in each case is a subset of the class of CRNs respecting finite density) with $k$ species and initial configuration $\vi \in \N^k$ has the property that every producible species is producible from initial configuration $m \cdot \vi$ through at most $t$ \emph{reactions}, where $m$ and $t$ are bounded by a polynomial in $k$.
In other words, if the goal of the CRN is to delay the production of some species until $\omega(\poly(k))$ reactions have occurred (such as the ``Grey code counter'' CRN of~\cite{CondonHMT12}, which iterates through $2^k$ states before producing the first copy of a certain species), then the CRN cannot be \emph{multi-copy tolerant}, since $m$ copies of the CRN running in parallel can ``short-circuit'' and produce every species in $\poly(k)$ reactions, if the reactions are carefully chosen.
Our main theorem extends these deterministic impossibility results to the stochastic model, showing that not only is there a short sequence of reactions that will produce every species,
but furthermore that the laws of chemical kinetics will force such a sequence to actually happen in constant time with high probability.

The paper is organized as follows.
Section~\ref{sec-prelim} defines the model of stochastic chemical reaction networks.
Section~\ref{sec-timer} states and proves the main theorem, Theorem~\ref{thm-no-timer}, that every CRN respecting finite density, starting from a dense initial configuration, likely produces every producible species in constant time.
Appendix~\ref{sec-exponential-decay} proves a Chernoff bound on continuous-time stochastic exponential decay processes, Lemma~\ref{lem-exp-decay-chernoff}, used in the proof of Theorem~\ref{thm-no-timer}.
Appendix~\ref{sec-random-walks} proves a Chernoff bound on continuous-time biased random walks, Lemma~\ref{lem-random-walk-f-r}, a messy-to-state consequence of which (Lemma~\ref{lem-random-walk-reflecting}) is used in the proof of Theorem~\ref{thm-no-timer}.
Section~\ref{sec-conclusion} discusses questions for future work.


\section{Preliminaries}
\label{sec-prelim}


\subsection{Chemical reaction networks}

If $\Lambda$ is a finite set of chemical species,
we write $\N^\Lambda$ to denote the set of functions $f:\Lambda \to \N$.
Equivalently, we view an element $\vc\in\N^\Lambda$ as a vector of $|\Lambda|$ nonnegative integers, with each coordinate ``labeled'' by an element of $\Lambda$.
Given $S\in \Lambda$ and $\vc \in \N^\Lambda$, we refer to $\vc(S)$ as the \emph{count of $S$ in $\vc$}.
Let $\|\vc\| = \|\vc\|_1 = \sum_{S \in \Lambda} \vc(S)$ represent the total count of species in $\vc$.
We write $\vc \leq \vc'$ to denote that $\vc(S) \leq \vc'(S)$ for all $S \in \Lambda$.
Given $\vc,\vc' \in \N^\Lambda$, we define the vector component-wise operations of addition $\vc+\vc'$, subtraction $\vc-\vc'$, and scalar multiplication $n \vc$ for $n \in \N$.
If $\Delta \subset \Lambda$, we view a vector $\vc \in \N^\Delta$ equivalently as a vector $\vc \in \N^\Lambda$ by assuming $\vc(S)=0$ for all $S \in \Lambda \setminus \Delta.$
For all $\vc \in \N^\Lambda$, let $[\vc] = \setr{S \in \Lambda}{\vc(S) > 0}$ be the set of species with positive counts in $\vc$.

Given a finite set of chemical species $\Lambda$, a \emph{reaction} over $\Lambda$ is a triple $\beta = ( \vr,\vp,k ) \in \N^\Lambda \times \N^\Lambda \times \R^+$, specifying the stoichiometry of the reactants and products, respectively, and the \emph{rate constant} $k$.
For instance, given $\Lambda=\{A,B,C\}$, the reaction $A+2B \overset{4.7}{\to} A+3C$ is the triple $((1,2,0),(1,0,3),4.7).$
A \emph{(finite) chemical reaction network (CRN)} is a pair $\calC=(\Lambda,R)$, where $\Lambda$ is a finite set of chemical \emph{species}, and $R$ is a finite set of reactions over $\Lambda$.
A \emph{configuration} of a CRN $\calC=(\Lambda,R)$ is a vector $\vc \in \N^\Lambda$.
When the configuration $\vc$ is clear from context, we write $\# X$ to denote $\vc(X)$. 

Given a configuration $\vc$ and reaction $\beta=(\vr,\vp,k)$, we say that $\beta$ is \emph{applicable} to $\vc$ if $\bfr \leq \vc$ (i.e., $\vc$ contains enough of each of the reactants for the reaction to occur).
If $\beta$ is applicable to $\vc$, then write $\beta(\vc)$ to denote the configuration $\vc + \bfp - \bfr$ (i.e., the configuration that results from applying reaction $\beta$ to $\vc$).
If $\vc'=\beta(\vc)$ for some reaction $\beta \in R$, we write $\vc \to_\calC \vc'$, or merely $\vc \to \vc'$ when $\calC$ is clear from context.
An \emph{execution} (a.k.a., \emph{execution sequence}) $\calE$ is a finite or infinite sequence of one or more configurations $\calE = (\vc_0, \vc_1, \vc_2, \ldots)$ such that, for all $i \in \{1,\ldots,|\calE|\}$, $\vc_{i-1} \to \vc_{i}$.
If a finite execution sequence starts with $\vc$ and ends with $\vc'$, we write $\vc \to_\calC^* \vc'$, or merely $\vc \to^* \vc'$ when the CRN $\calC$ is clear from context.
In this case, we say that $\vc'$ is \emph{reachable} from $\vc$.

We say that a reaction $\beta = ( \vr,\vp,k )$ \emph{produces} $S \in \Lambda$ if $\vp(S) - \vr(S) > 0$, and that $\beta$ \emph{consumes} $S$ if $\vr(S) - \vp(S) > 0$.
When a CRN $\calC=(\Lambda,R)$ and an initial configuration $\vi \in \N^\Lambda$ are clear from context, we say a species $S \in \Lambda$ is \emph{producible} if there exists $\vc$ such that $\vi \to_\calC^* \vc$ and $\vc(S)>0$.

\subsection{Kinetic model}

The following model of stochastic chemical kinetics is widely used in quantitative biology and other fields dealing with chemical reactions in which stochastic effects are significant~\cite{Gillespie77}.
It ascribes probabilities to execution sequences, and also defines the time of reactions, allowing us to study the running time of CRNs in Section~\ref{sec-timer}.

A reaction is \emph{unimolecular} if it has one reactant and \emph{bimolecular} if it has two reactants.
We assume no higher-order reactions occur.\footnote{This assumption is not critical to the proof of the main theorem; the bounds derived on reaction rates hold asymptotically even higher-order reactions are permitted, although the constants would change.} 
The kinetic behavior of a CRN is described by a continuous-time Markov process, whose states are configurations of the CRN, as follows.
Given a fixed volume $v > 0$ and current configuration $\vc$, the \emph{propensity} of a unimolecular reaction $\beta : X \overset{k}{\to} \ldots$ in configuration $\vc$ is $\rho(\vc, \beta) = k \vc(X)$.
The propensity of a bimolecular reaction $\beta : X + Y \overset{k}{\to} \ldots$, where $X \neq Y$, is $\rho(\vc, \beta) = \frac{k}{v} \vc(X) \vc(Y)$.
The propensity of a bimolecular reaction $\beta : X + X \overset{k}{\to} \ldots$ is $\rho(\vc, \beta) = \frac{k}{v} \frac{\vc(X) (\vc(X) - 1)}{2}$.\footnote{Intuitively, $\rho(\vc,\beta)$ without the $k$ and $v$ terms counts the number of ways that reactants can collide in order to react (with unimolecular reactants ``colliding'' with only themselves).}
The propensity function determines the evolution of the system as follows.
The time until the next reaction occurs is an exponential random variable with rate $\rho(\vc) = \sum_{\beta \in R} \rho(\vc,\beta)$. 
The probability that next reaction will be a particular $\beta_{\text{next}}$ is $\frac{\rho(\vc,\beta_{\text{next}})}{\rho(\vc)}$.


\section{CRNs that produce all species in constant time}
\label{sec-timer}
We say a CRN \emph{respects finite density} if there is a constant $\hat{c}$ such that, for any initial configuration $\vi$ and any configuration $\vc$ such that $\vi \to^* \vc$, $\|\vc\| \leq \hat{c} \|\vi\|$.
That is, the maximum molecular count attainable is bounded by a constant multiplicative factor of the initial counts, implying that if the volume is sufficiently large to contain the initial molecules, then within a constant factor, it is sufficiently large to contain all the molecules ever produced.
We therefore safely assume that for any CRN respecting finite density with initial configuration $\vi$, the volume is $\|\vi\|$.

CRNs respecting finite density constitute a wide class of CRNs that includes, for instance, all mass-conserving CRNs (CRNs for which there exists a mass function $m:\Lambda\to\R^+$ such that $\sum_{S\in\vr} m(S) = \sum_{S\in\vp} m(S)$ for all reactions $( \vr,\vp,k )$).

In this section we prove that no CRN respecting finite density, starting from an initial configuration with ``large'' species counts, can delay the production of any producible species for more than a constant amount of time.
All producible species are produced in time $O(1)$ with high probability.


Let $\alpha > 0$.
We say that a configuration $\vc$ is \emph{$\alpha$-dense} if for all $S \in \Lambda$, $S \in [\vc] \implies \vc(S) \geq \alpha \|\vc\|$.
Given initial configuration $\vi \in \N^\Lambda$, let $\Lambda^*_\vi = \setr{S \in \Lambda}{(\exists \vc)\ \vi \to^* \vc \text{ and } \vc(S) > 0}$ denote the set of species producible from $\vi$.
For all $t > 0$ and $S \in \Lambda$, let $\#_t S$ be the random
  variable representing the count of $S$ after $t$ seconds, if $\vi$ is the initial configuration at time 0. 

The following is the main theorem of this paper.

  \newcommand{\cmax}{\hat{c}}
  \newcommand{\kmax}{\hat{K}}
  \newcommand{\kmin}{\hat{k}}

\begin{thm}\label{thm-no-timer}
  Let $\alpha > 0$ and $\calC = (\Lambda,R)$ be a CRN respecting finite density.
  Then there are constants $\epsilon,\delta,t > 0$ such that, for all sufficiently large $n$ (how large depending only on $\alpha$ and $\calC$), for all $\alpha$-dense initial configurations $\vi$ with $||\vi|| = n$,
  $
    \Pr[(\forall S\in\Lambda^*_\vi)\ \#_t S \geq \delta n] \geq 1 - 2^{-\epsilon n}.
  $
\end{thm}

\begin{proof}
  The handwaving intuition of the proof is as follows.
  We first show that every species initially present, because they have count $\Omega(n)$ and decay at an at most exponential rate, remain at count $\Omega(n)$ (for some smaller constant fraction of $n$) for a constant amount of time.
  Because of this, during this entire time, the reactions for which they are reactants are running at rate $\Omega(n)$.
  Therefore the products of those reactions are produced quickly enough to get to count $\Omega(n)$ in time $O(1)$.
  Because \emph{those} species (the products) decay at an at most exponential rate, they have count $\Omega(n)$ for a constant amount of time.
  Therefore the reactions for which they are \emph{reactants} are running at rate $\Omega(n)$.
  Since there are a constant number of reactions, this will show that all producible species are produced in constant time with high probability.
  The tricky part is to show that although species may be consumed and go to lower counts, every species produced has count $\Omega(n)$ for $\Omega(1)$ time, sufficient for it to execute $\Omega(n)$ reactions of which it is a reactant and produce $\Omega(n)$ of the products of that reaction.

  On to the gritty details.
  Because $\calC$ respects finite density, the maximum count obtained by any species is $O(n)$.
  Let $\cmax$ be a constant such that all species have count at most $\cmax n$ in any configuration reachable from $\vi$.
  Let $\kmax = \sum_{(\vr,\vp,k) \in R} k$ be the sum of the rate constants of every reaction.
  Let $\kmin = \min_{(\vr,\vp,k) \in R} k$ be the minimum rate constant of any reaction.
  For any $\Delta \subseteq \Lambda$, define
  $$
    \mathrm{PROD}(\Delta) =  \setr{S \in \Lambda}{(\exists \beta=( \vr,\vp,k ) \in R)\ \beta \text{ produces $S$ and } [\vr] \subseteq \Delta}
  $$
  to be the set of species producible by a single reaction assuming that only species in $\Delta$ are present, and further assuming that those species have sufficient counts to execute the reaction.

  Define subsets $\Lambda_0 \subset \Lambda_1 \subset \ldots \subset \Lambda_{m-1} \subset \Lambda_m \subseteq \Lambda$ as follows.
  Let $\Lambda_0 = [\vi]$.
  For all $i \in \Z^+$, define
  $
    \Lambda_i = \Lambda_{i-1} \cup \mathrm{PROD}(\Lambda_{i-1}).
  $
  Let $m$ be the smallest integer such that $\mathrm{PROD}(\Lambda_m) \subseteq \Lambda_m$.
  Therefore for all $i \in \{1,\ldots,m\}$, $|\Lambda_i| > |\Lambda_{i-1}|$, whence $m < |\Lambda|$.

  By the hypothesis of the theorem, all $S \in \Lambda_0$ satisfy $\vi(S) \geq \alpha n$.
  $S$ may be produced and consumed.
  To establish that $\# S$ remains high for a constant amount of time, in the worst case we assume that $S$ is only consumed.
  Let $\lambda = \cmax \kmax$.
  We will assume that $\lambda \geq 1$ since we can choose $\cmax \geq \frac{1}{\kmax}$ if it is not already greater than $\frac{1}{\kmax}$.
  This assumption is required in Lemma~\ref{lem-random-walk-reflecting}, which is employed later in the proof.
  We also assume that $\cmax \geq 1$, implying $\lambda \geq \kmax$.

  For all $S \in \Lambda$, the rate of any unimolecular reaction consuming $S$ is at most $\kmax \# S \leq \lambda \# S$, and the rate of any bimolecular reaction consuming $S$ is at most $\frac{\kmax}{n} (\cmax n) \# S = \lambda \# S$, whence the expected time for any reaction consuming $S$ is at least $\frac{1}{\lambda \# S}$.
  We can thus upper bound the overall consumption of $S$ by an exponential decay process with rate $\lambda$, i.e., this process will consume copies of $S$ at least as quickly as the actual CRN consumes copies of $S$.

  Let $c = 4 e^{\lambda (m+1)}$. Let $\delta_0 = \frac{\alpha}{c}$.
  Fix a particular $S \in \Lambda_0$.
  Let $N = \alpha n$ and $\delta = \frac{1}{c}$ in Lemma~\ref{lem-exp-decay-chernoff}.
  Then by the fact that the rate of decay of $S$ is bounded by an exponential decay process with rate $\lambda$ and initial value $\alpha n$ (the process $\bD_\lambda^{N}$ as defined in Section~\ref{sec-exponential-decay}) and Lemma~\ref{lem-exp-decay-chernoff},
  \begin{align*}
    \Pr[(\exists t \in [0,m+1])\ \#_{t} S < \delta_0 n]
    \leq
    \Pr \left[ \bD_\lambda^{\alpha n}(m+1) < \frac{1}{c} \alpha n \right]
    <
    \left( 2 \frac{1}{c} e^{\lambda (m+1)} \right)^{\alpha n / c - 1}
    =
    2^{- \delta_0 n + 1}.
  \end{align*}
  By the union bound,
  \begin{equation}\label{ineq-Lambda0-high}
    \Pr[(\exists S \in \Lambda_0)(\exists t \in [0,m+1])\ \#_t S < \delta_0 n]
    <
    |\Lambda_0| 2^{- \delta_0 n + 1}.
  \end{equation}

  That is, with high probability, all species in $\Lambda_0$ have ``high'' count (at least $\delta_0 n$) for the entire first $m+1$ seconds.
  Call this event $H(\Lambda_0)$ (i.e., the complement of the event in \eqref{ineq-Lambda0-high}).

  We complete the proof by a ``probabilistic induction''\footnote{In other words, we show that if the induction hypothesis fails with low probability $p$, and if the induction step fails with low probability $q$, given that the induction hypothesis holds, then by the union bound, the induction step fails with probability at most $p+q$.} on $i \in \{0,1,\ldots,m\}$ as follows.
  Inductively assume that for all $X \in \Lambda_i$ and all $t \in [i,m+1]$, $\#_t X \geq \delta_i n$ for some $\delta_i > 0$.
  Call this event $H(\Lambda_i)$.
  Then we will show that for all $S \in \Lambda_{i+1}$ and all $t \in [i+1,m+1]$, assuming $H(\Lambda_i)$ is true, with high probability $\#_t S \geq \delta_{i+1} n$ for some $\delta_{i+1} > 0$.
  We will use Lemma~\ref{lem-random-walk-reflecting} to choose particular values for the $\delta_i$'s, and these values will not depend on $n$.

  The base case is established by~\eqref{ineq-Lambda0-high} for $\delta_0 = \frac{\alpha}{c}$.
  Fix a particular species $S \in \Lambda_{i+1}$.
  By the definition of $\Lambda_{i+1}$, it is produced by either at least one reaction of the form $X \to S + \ldots$ for some $X \in \Lambda_i$ or by $X+Y \to S + \ldots$ for some $X,Y \in \Lambda_i$.
  By the induction hypothesis $H(\Lambda_i)$, for all $t \in [i,m+1]$, $\#_t X \geq \delta_i n$ and $\#_t Y \geq \delta_i n$ for some $\delta_i > 0$.
  In the case of a unimolecular reaction, the propensity of this reaction is at least $\kmin \delta_i n$.
  In the case of a bimolecular reaction, the propensity is at least $\frac{\kmin}{n} (\delta_i n)^2 = \kmin \delta_i^2 n$ if $X \neq Y$ or $\frac{\kmin}{n} \frac{\delta_i n (\delta_i n - 1)}{2} = \frac{\kmin}{2}  (\delta_i^2 n - 1)$ if $X = Y$.

  The last of these three possibilities is the worst case (i.e., the smallest).
  Because $\delta_i$ is a constant independent of $n$, for sufficiently large $n$ ($n \geq 2 / \delta_i^2 $), $\delta_i^2 n - 1 \geq \delta_i^2 n / 2$, whence the calculated rate is at least $\frac{\kmin}{4}  \delta_i^2 n$.

  Let $\delta_f = \frac{\kmin \delta_i^2}{4}$ and $\delta_r = \frac{\delta_f}{4 \lambda}$.
  Similar to the argument above concerning the maximum rate of decay of any $S \in \Lambda_0$, the rate at which reactions consuming $S \in \Lambda_{i+1}$ occur is at most $\lambda \# S$.
  Also, by the above arguments, the minimum rate at which reactions producing $S$ proceed is at least $\delta_f n$.
  Therefore we may lower bound the \emph{net} production of $S$ (i.e., its total count) by a continuous-time random walk on $\N$ that starts at 0, has constant forward rate $\delta_f n$ from every state $i$ to $i+1$, and has reverse rate $\lambda i$ from $i$ to $i-1$, defined as the process $\bW^n_{\delta_f,\lambda}$ in Section~\ref{sec-random-walks}, Lemma~\ref{lem-random-walk-reflecting}.
  By this lower bound and Lemma~\ref{lem-random-walk-reflecting}\footnote{Although Lemma~\ref{lem-random-walk-reflecting} is stated for times between 0 and 1, here we use it for times between $i$ and $i+1$ (just shift all times down by $i$).} (recall that we have assumed $\lambda \geq 1$),
  \begin{align}
    \Pr \left[ \max_{\hatt \in [i,i+1]} \#_{\hatt} S < \delta_r n \right]
    \leq
    \Pr \left[ \max_{\hatt \in [0,1]} \bW^N_{\delta_f,\lambda}(\hatt) < \delta_r n \right] 
    <
    2^{-\delta_f n / 22 + 1}
    =
    2^{ - \kmin \delta_i^2 n / 88 + 1}.   \label{ineq-S-prod}
  \end{align}

  Therefore with high probability $\# S$ reaches count at least $\delta_r n$ at some time $\hatt \in [i,i+1]$.
  Let
  $
    \delta_{i+1}
    =
    \frac{\delta_r}{c}.
  $
  As before, we can bound the decay rate of $S$ from time $\hatt$ until time $m+1$ by an exponential decay process with rate $\lambda$ and initial value $\delta_r n$, proceeding for $m+1$ seconds (since $\hatt \geq 0$).
  By Lemma~\ref{lem-exp-decay-chernoff},
  \begin{align}
    \Pr[(\exists t \in [\hatt,m+1])\ \#_t S < \delta_{i+1} n]
    \leq
    \Pr \left[ \bD_\lambda^{\delta_r n}(m+1) < \frac{\lambda}{c} \delta_r n \right]
    <
    \left( 2 \frac{1}{c} e^{\lambda (m+1)} \right)^{\lambda \delta_r n / c - 1}
    =
    2^{- \delta_{i+1} n + 1}. \label{ineq-S-not-decay}
  \end{align}

  Fix a particular $S \in \Lambda_{i+1}$.
  By the union bound applied to~\eqref{ineq-S-prod} and~\eqref{ineq-S-not-decay} and the fact that $\hatt \leq i+1$, the probability that $\#_t S < \delta_{i+1} n$ at any time $t \in [i+1,m+1]$, given that the induction hypothesis $H(\Lambda_i)$ holds (i.e., given that $\#_t X \geq \delta_i n$ for all $X \in \Lambda_i$ and all $t \in [i,m+1]$), is at most $2^{ - \kmin \delta_i^2 n / 88 + 1} + 2^{- \delta_{i+1} n + 1}$.
  Define $H(\Lambda_{i+1})$ to be the event that this does not happen for \emph{any} $S \in \Lambda_{i+1}$, i.e., that $(\forall S \in \Lambda_{i+1}) (\forall t \in [i,m+1])\ \#_t S \geq \delta_{i+1} n$.
  By the union bound over all $S \in \Lambda_{i+1}$,
  \begin{equation} \label{ineq-Lambdai-high}
    \Pr \left[ \neg H(\Lambda_{i+1}) | H(\Lambda_i) \right]
    <
    |\Lambda_{i+1}| \left( 2^{ - \kmin \delta_i^2 n / 88 + 1} + 2^{- \delta_{i+1} n + 1} \right)
  \end{equation}

  By the union bound applied to~\eqref{ineq-Lambda0-high} and~\eqref{ineq-Lambdai-high}, and the fact that $\delta_{m} < \delta_{i}$ for all $i\in\{0,\ldots,m-1\}$, the probability that any step of the induction fails is at most
  \begin{align}
    \Pr[\neg H(\Lambda_0)]
    +
    \sum_{i=1}^m \Pr[\neg H(\Lambda_i) | H(\Lambda_{i-1})]
    &<
    |\Lambda_0| 2^{- \delta_0 n + 1}
    +
    \sum_{i=1}^m |\Lambda_i| \left( 2^{ - \kmin \delta_{i-1}^2 n / 88 + 1} + 2^{- \delta_i n + 1} \right) \nonumber
    \\&<
    |\Lambda| \left( 2^{- \delta_m n + 1} + 2^{ - \kmin \delta_{m}^2 n / 88 + 1} \right) 
    <
    |\Lambda| \left( 2^{ - \kmin \delta_{m}^2 n / 88 + 2} \right). \label{ineq-prob-fail}
  \end{align}

  At this point we are essentially done.
  The remainder of the proof justifies that the various constants involved can be combined into a single constant $\epsilon$ that does not depend on $n$ (although it depends on $\calC$ and $\alpha$) as in the statement of the theorem.

  By our choice of $\delta_{i+1}$, we have
  $$
    \delta_{i+1}
    =
    \frac{\delta_r}{c}
    =
    \frac{\delta_f}{4 \lambda c}
    =
    \frac{\delta_i^2 \kmin}{16 \lambda c}
    >
    \left( \frac{\delta_i \kmin}{16 \lambda c} \right)^2.
  $$

  Recall that $\delta_0 = \frac{\alpha}{c}$.
  Therefore, for all $i \in \{1,\ldots,m\}$,
  $
  \delta_{i}
  > \left( \frac{\alpha \kmin}{16 \lambda c^2} \right)^{2^{i}}.
  $
  So
  \begin{equation}\label{ineq-delta-m}
  \delta_m
  > \left( \frac{\alpha \kmin}{16 \lambda c^2} \right)^{2^m}
  \geq \left( \frac{\alpha \kmin}{16 \lambda c^2} \right)^{2^{|\Lambda|-1}}.
  \end{equation}
  Therefore, if $n > \left( \frac{\alpha \kmin}{16 \lambda c^2} \right)^{2^{|\Lambda|-1}}$, then by~\eqref{ineq-prob-fail} and~\eqref{ineq-delta-m}, the failure probability in~\eqref{ineq-prob-fail} is at most
  \begin{align*}
    |\Lambda| 2^{ - \kmin \delta_{m}^2 n / 88 + 2}
    &<
    |\Lambda| 2^{ - \kmin \left( \left( \frac{\alpha \kmin}{16 \lambda c^2} \right)^{2^{|\Lambda|-1}} \right)^2 n / 88 + 2}
    &&=
    |\Lambda| 2^{ - \kmin \left( \frac{\alpha \kmin}{16 \lambda c^2} \right)^{2^{|\Lambda|}} n / 88 + 2}
    \\=
    |\Lambda| 2^{ - \kmin \left( \frac{\alpha \kmin}{16 \lambda \left( 4 e^{\lambda (m+1)} \right)^2} \right)^{2^{|\Lambda|}} n / 88 + 2}
    &=
    |\Lambda| 2^{ - \kmin \left( \frac{\alpha \kmin}{256 \lambda e^{2 \lambda (m+1)}} \right)^{2^{|\Lambda|}} n / 88 + 2}
    &&\leq
    |\Lambda| 2^{ - \kmin \left( \frac{\alpha \kmin}{256 \kmax \cmax e^{2 \kmax \cmax |\Lambda|} } \right)^{2^{|\Lambda|}} n / 88 + 2}.
  \end{align*}

  Letting $\epsilon' = \frac{\kmin}{88} \left( \frac{\alpha \kmin}{256 \kmax \cmax e^{2 \kmax \cmax |\Lambda|}} \right)^{2^{|\Lambda|}}$ implies failure probability at most $|\Lambda| 2^{-\epsilon' n + 2}$.
  For $n \geq \frac{2 + 2 \log |\Lambda|}{\epsilon'}$, $|\Lambda| 2^{-\epsilon' n + 2} \leq 2^{-(\epsilon' / 2) n}$.
  Letting $t=m+1$, $\delta = \delta_m$, and $\epsilon = \epsilon' / 2$ completes the proof.
\end{proof}

We have chosen $t=m+1$ (i.e., $t$ depends on $\calC$); however, the same proof technique works if we choose $t=1$ (or any other constant independent of $\calC$), although this results in smaller values for $\delta$ and $\epsilon$.

Although the choice of $\epsilon$ is very small, the analysis used many very loose bounds for the sake of simplifying the argument.
A more careful analysis would show that a much larger value of $\epsilon$ could be chosen, but for our purposes it suffices that $\epsilon$ depends on $\calC$ and $\alpha$ but not on the volume $n$.

However, it does seem fundamental to the analysis that $\epsilon \leq \gamma^{2^{|\Lambda|}}$ for some $0 < \gamma < 1$.
Consider the CRN with $n$ initial copies of $X_1$ in volume $n$ and reactions
\[
  \begin{array}{rclrclcrcl}
      X_1 &\to& \emptyset, & X_2 &\to& \emptyset, &\ldots& X_{m} &\to& \emptyset, \\
      X_1 + X_1 &\to& X_2, & X_2 + X_2 &\to& X_3, &\ldots& X_{m} + X_{m} &\to& X_{m+1}
  \end{array}
\]
This is a CRN in which the number of stages $m$ is actually equal to its worst-case value $|\Lambda|-1$, and in which each species is being produced at the minimum rate possible -- by a bimolecular reaction -- and consumed at the fastest rate possible -- by a unimolecular reaction (in addition to a bimolecular reaction).
Therefore extremely large values of $n$ are required for the ``constant fraction of $n$'' counts in the proof to be large enough to work, i.e., for the production reactions to outrun the consumption reactions.
This appears to be confirmed in stochastic simulations as well.

\section{Conclusion}
\label{sec-conclusion}


The reason we restrict attention to CRNs respecting finite density is that the proof of Theorem~\ref{thm-no-timer} relies on bounding the rate of consumption of any species by an exponential decay process, i.e., no species $S$ is consumed faster than rate $\lambda \# S$, where $\lambda > 0$ is a constant.
This is not true if the CRN is allowed to violate finite density, because a reaction consuming $S$ such as $X + S \to X$ proceeds at rate $\# X \# S / n$ in volume $n$, and if $\# X$ is an arbitrarily large function of the volume, then for any constant $\lambda$, this rate will eventually exceed $\lambda \#S$.

In fact, for the proof to work, the CRN need not respect finite density for all time, but only for some constant time after $t=0$.
This requirement is fulfilled even if the CRN has non-mass-conserving reactions such as $X \to 2X$.
Although such a CRN eventually violates finite density, for all constant times $t>0$, there is a constant $c > 0$ such that $\#_t X \leq c \#_0 X$ with high probability.

The assumption that the CRN respects finite density is a perfectly realistic constraint satisfied by all real chemicals,\footnote{The reader may notice that recent work has proposed concrete chemical implementations of arbitrary CRNs using DNA strand displacement~\cite{SolSeeWin10,cardelli2011strand}. However, these implementations assume a large supply of ``fuel'' DNA complexes supplying mass and energy for reactions such as $X \to 2X$ that violate the laws of conservation of mass and energy.} and as noted, many non-mass-conserving CRNs respect finite density for a constant amount of initial time.
Nevertheless, it is interesting to note that there are syntactically correct CRNs that violate even this seemingly mild constraint.
Consider the reaction $2X \to 3X$.
This reaction has the property that with high probability, $\# X$ goes to infinity in constant time.
Therefore it is conceivable that such a reaction could be used to consume a species $Y$ via $X + Y \to X$ at such a fast rate that, although species $S$ is producible via $Y \to S$, with high probability $S$ is never produced, because copies of $Y$ are consumed quickly by $X$ before they can undergo the unimolecular reaction that produces $S$.

We have not explored this issue further since such CRNs are not physically implementable by chemicals.
However, it would be interesting to know whether such a CRN violating finite density could be used to construct a counterexample to Theorem~\ref{thm-no-timer}.

An open problem is to precisely characterize the class of CRNs obeying Theorem~\ref{thm-no-timer}; as noted, those respecting finite density with dense initial configurations are a strict subset of this class.

\paragraph{Acknowledgements.}
The author is very grateful to Chris Thachuk, Anne Condon, Damien Woods, Manoj Gopalkrishnan, David Soloveichik, David Anderson, and Lea Popovic for many insightful discussions.

\bibliographystyle{plain}
\bibliography{tam}


\newpage
\appendix
{\LARGE \bf \noindent Technical appendix}

\section{Chernoff bound for exponential decay}
\label{sec-exponential-decay}
In this section, we prove a Chernoff bound on the probability that an exponential decay process reduces to a constant fraction of its initial value after some amount of time.

For all $N\in\Z^+$ and $\lambda>0$, let $\bD_\lambda^N(t)$ be a Markov process on $\{0,1,\ldots,N\}$ governed by exponential decay, with initial value $N$ and decay constant $\lambda$; i.e., $\bD_\lambda^N(t) = N$ for all $t \in[0,\bT_1)$, $\bD_\lambda^N(t)=N-1$ for all $t\in[\bT_1,\bT_1 + \bT_2)$, $\bD_\lambda^N(t)=N-2$ for all $t\in[\bT_1+\bT_2,\bT_1 + \bT_2+\bT_3)$, etc., where  for $i \in \{1,\ldots,N\}$, $\bT_i$ is an exponential random variable with rate $\lambda(N-i+1)$.

\begin{lem}\label{lem-exp-decay-chernoff}
  Let $N \in \Z^+$, $\lambda,t > 0$, and $0 < \delta < 1$.
  Then
  $$\Pr[\bD_\lambda^N(t) < \delta N] < (2 \delta e^{\lambda t})^{\delta N-1}.$$
\end{lem}

\begin{proof}
  Let $\epsilon$ be the largest number such that $\epsilon < \delta$ and $\epsilon N \in \N$.
  Then $\epsilon N < \delta N \leq \epsilon N + 1$.

  Define $\bT^\delta = \sum_{i=1}^{N - \epsilon N} \bT_i$ as the random variable representing the time required for $N- \epsilon N$ decay events to happen, i.e., for $\bD_\lambda^N$ to decay to strictly less than $\delta N$ from its initial value $N$.

  If $\bX$ is an exponential random variable with rate $\lambda'$, then the moment-generating function $M_\bX:\R\to\R$ of $\bX$ is $M_\bX(\theta) = \E[e^{\theta \bX}] = \frac{\lambda'}{\lambda' - \theta}$, defined whenever $|\theta| < \lambda'$~\cite{chitale2008random}.
  Therefore
  $$
    M_{\bT_i}(\theta) 
    = \frac{\lambda(N-i+1)}{\lambda(N-i+1)-\theta} = \frac{N-i+1}{N-i+1-\frac{\theta}{\lambda}},
  $$
  defined whenever $|\theta|  < \lambda(N-i+1)$.
  In particular, the smallest such value is $\lambda (N - (N - \epsilon N) + 1) = \lambda (\epsilon N + 1)$, so we must choose $|\theta| < \lambda (\epsilon N + 1)$.

  Because each $\bT_i$ is independent, the moment-generating function of $\bT^\delta$ is
  \begin{eqnarray*}
    M_{\bT^\delta}(\theta)
    &=& \E[e^{\theta \bT^\delta}]
    = \E[e^{\theta \sum_{i=1}^{N - \epsilon N} \bT_i}]
    = \E \left[ \prod_{i=1}^{N-\epsilon N} e^{\theta \bT_i} \right]
    \\&=& \prod_{i=1}^{N-\epsilon N} \E[e^{\theta \bT_i}]
    = \prod_{i=1}^{N-\epsilon N} \frac{N-i+1}{N-i+1-\frac{\theta}{\lambda}}
    = \prod_{i=\epsilon N + 1}^{N} \frac{i}{i-\frac{\theta}{\lambda}}.
  \end{eqnarray*}
  For any $t>0$, the event that $\bT^\delta < t$ (it takes fewer than $t$ seconds for $N - \epsilon N$ decay events to happen) is equivalent to the event that $\bD_\lambda^N(t) < \delta N$ ($\bD_\lambda^N$, which started at $\bD_\lambda^N(0)=N$, has experienced at least $N - \epsilon N$ decay events after $t$ seconds to arrive at $\bD_\lambda^N(t) < \delta N$).

  Using Markov's inequality, for all $\theta < 0$,
  \begin{eqnarray*}
    \Pr[\bD_\lambda^N(t) < \delta N]
    =
    \Pr[\bT^\delta < t]
    =
    \Pr[e^{\theta \bT^\delta} > e^{\theta t}]
    \leq
    \frac{\E[e^{\theta \bT^\delta}]}{e^{\theta t}}
    =
    \frac{1}{e^{\theta t}} \prod_{i=\epsilon N + 1}^{N} \frac{i}{i-\frac{\theta}{\lambda}}
  \end{eqnarray*}

  Let $\theta = -\lambda \epsilon N$.
  This implies $|\theta| = \lambda \epsilon N < \lambda (\epsilon N + 1)$ as required for $M_{\bT_i}(\theta)$ to be defined for all $i \in \{1,\ldots,N-\epsilon N\}$.

  Then
  \begin{eqnarray*}
    &&\Pr[\bD_\lambda^N(t) < \delta N]
    \leq
    e^{\lambda \epsilon N t} \prod_{i=\epsilon N + 1}^{N} \frac{i}{i+\epsilon N}
    \\&=&
    e^{\lambda \epsilon N t} \left( \frac{\epsilon N}{\epsilon N + \epsilon N + 1} \right)
    \left( \frac{\epsilon N + 1}{\epsilon N + \epsilon N + 2} \right)
    \ldots
    \left( \frac{N - 1}{N + \epsilon N - 1} \right)
    \left( \frac{N}{N + \epsilon N} \right)
    \\&=&
    e^{\lambda \epsilon N t}
    \frac{(\epsilon N)(\epsilon N + 1)\ldots (\epsilon N + \epsilon N - 2) (\epsilon N + \epsilon N - 1)}{(N + 1)(N+2) \ldots (N+\epsilon N - 1) (N+\epsilon N)} \ \ \ \ \text{cancel terms}
    \\&<&
    e^{\lambda \epsilon N t}
    \frac{(2 \epsilon N)^{\epsilon N}}{N^{\epsilon N}}
    <
    (2 \epsilon e^{\lambda t})^{\epsilon N}
    <
    (2 \delta e^{\lambda t})^{\delta N - 1},
  \end{eqnarray*}
  which completes the proof.
\end{proof}

\section{Chernoff bounds for biased random walks}
\label{sec-random-walks}
This section proves Chernoff bounds for biased random walks that are used in Section~\ref{sec-timer}.
Lemma~\ref{lem-random-walk-f-r} shows that a random walk on $\Z$ with forward rate $\hf$ and reverse rate $\hr$ has a high probability to take at least $\Omega((\hf-\hr) t)$ net forward steps after $t$ seconds.
Lemma~\ref{lem-random-walk-reflecting} uses Lemma~\ref{lem-random-walk-f-r} to show that a random walk on $\N$ (i.e., with a reflecting barrier at state 0) with reverse rate in state $i$ proportional to $i$, has a high probability to reach state $j$ in time $t$, where $j$ is sufficiently small based on the backward rate and $t$.

We require the following Chernoff bound on Poisson distributions, due to Franceschetti, Dousse, Tse, and Thiran~\cite{FraDouTseThi07}.

\begin{thm}[\cite{FraDouTseThi07}]
  \label{thm-poisson-chernoff}
  Let $\bP(\lambda)$ be a Poisson random variable with rate $\lambda$.
  Then for all $n \in \N$
  $$
    \Pr \left[ \bP(\lambda) \geq n \right] \leq e^{-\lambda} \left( \frac{e \lambda}{n} \right)^n
  $$
  if $n > \lambda$ and
  $$
    \Pr \left[ \bP(\lambda) \leq n \right] \leq e^{-\lambda} \left( \frac{e \lambda}{n} \right)^n
  $$
  if $n < \lambda$.
\end{thm}

The following corollaries are used in the proof of Lemma~\ref{lem-random-walk-f-r}.

\begin{cor}\label{cor-poisson-chernoff-gamma-leq}
  Let $0 < \gamma < 1$. Then
  $
    \Pr[\bP(\lambda) \leq \gamma \lambda]
    \leq
    e^{-\lambda} \left( \frac{e \lambda}{\gamma \lambda} \right)^{\gamma \lambda}
    =
    \left( \frac{e^{1-\frac{1}{\gamma}}}{\gamma} \right)^{\gamma \lambda}.
  $
\end{cor}

\begin{cor}\label{cor-poisson-chernoff-gamma-geq}
  Let $\gamma > 1$. Then
  $
    \Pr[\bP(\lambda) \geq \gamma \lambda]
    \leq
    e^{-\lambda} \left( \frac{e \lambda}{\gamma \lambda} \right)^{\gamma \lambda}
    =
    \left( \frac{e^{1-\frac{1}{\gamma}}}{\gamma} \right)^{\gamma \lambda}.
  $
\end{cor}

\begin{figure}[t]
\centering
\centering
        \begin{subfigure}[b]{3in}
                \centering
                {\includegraphics[width=3in]{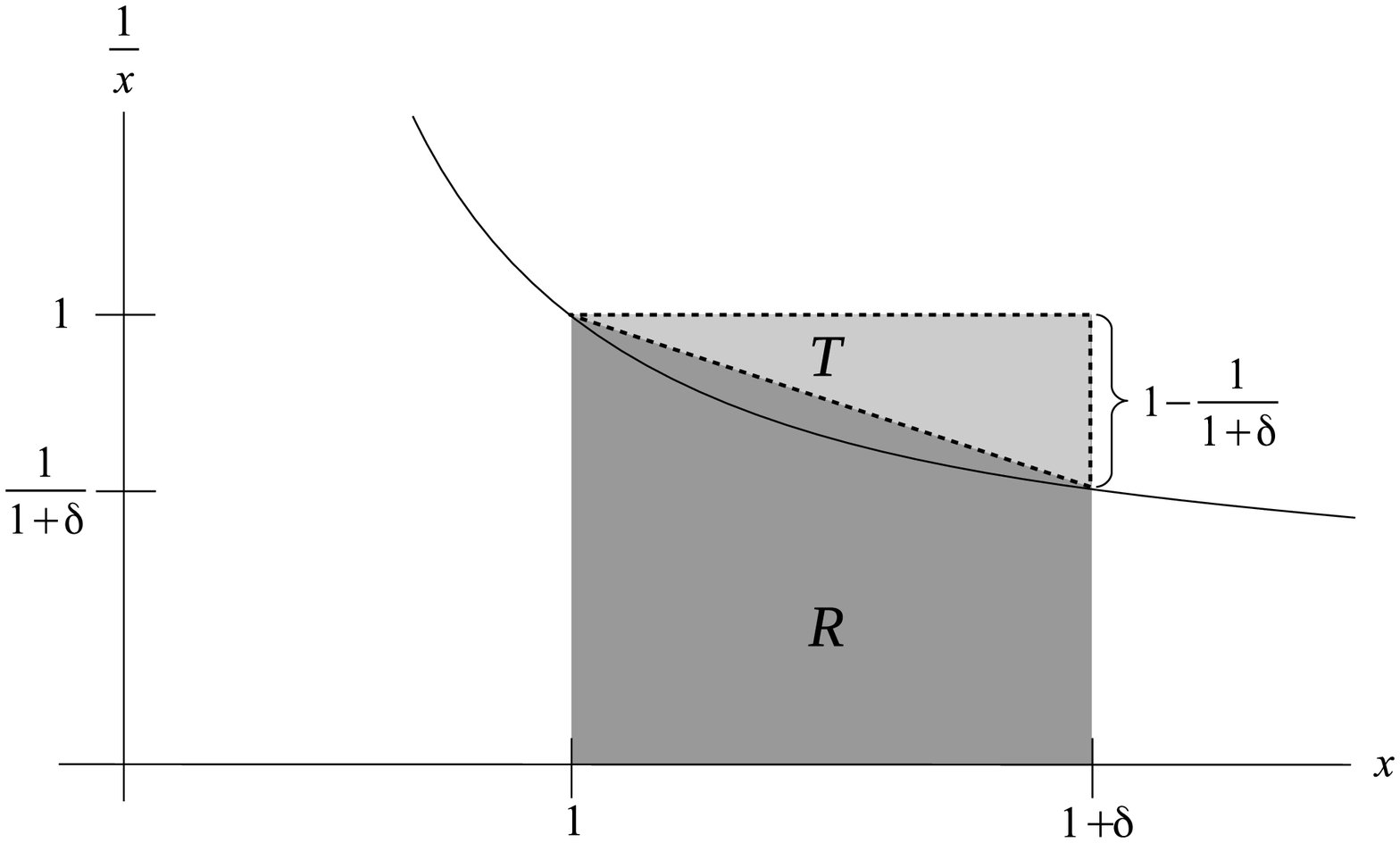}}
                \caption{\footnotesize Lemma~\ref{lem-ln-pos}: $\int_1^{1+\delta} \frac{1}{x}\ dx < \mathrm{area}(R) - \mathrm{area}(T)$.}
                \label{fig:ln-pos}
        \end{subfigure}
~ 
        \begin{subfigure}[b]{3in}
                \centering
                {\includegraphics[width=3in]{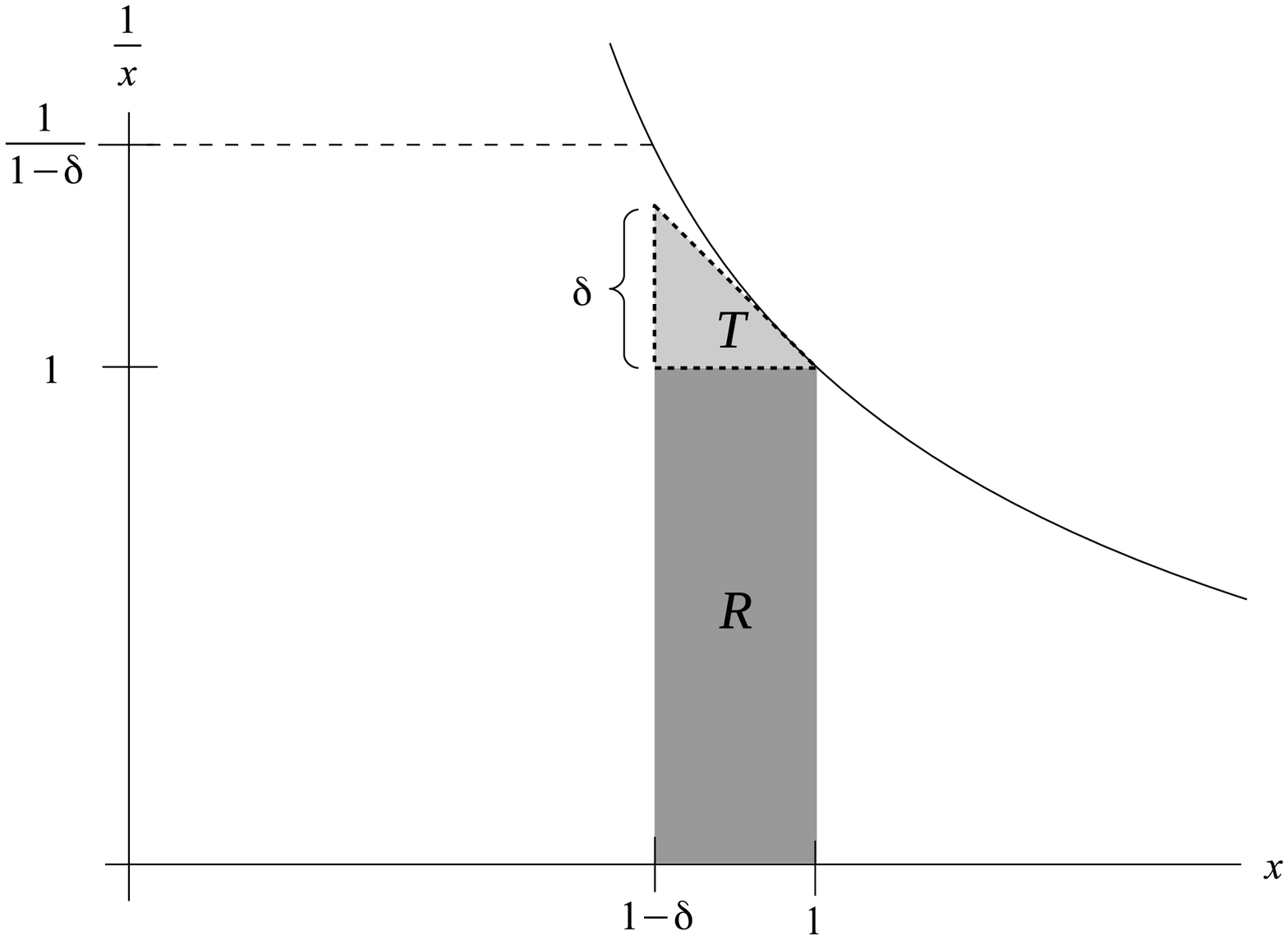}}
                \caption{\footnotesize Lemma~\ref{lem-ln-neg}: $\int_{1-\delta}^1 \frac{1}{x}\ dx > \mathrm{area}(R) + \mathrm{area}(T)$.}
                \label{fig:ln-neg}
        \end{subfigure}
\caption{
Illustration of Lemmas~\ref{lem-ln-pos} and~\ref{lem-ln-neg} (not to scale).
\label{fig:ln}
}
\end{figure}

We require the following bound on the natural logarithm function.

\begin{lem}\label{lem-ln-pos}
  Let $\delta > 0$.
  Then
  $$
    \ln (1+\delta) < \frac{\delta}{2} + \frac{\delta}{2(1+\delta)}.
  $$
\end{lem}

\begin{proof}
  See Figure~\ref{fig:ln-pos} for an illustration of the geometric intuition.
  Recall that for $a \geq 1$, $\ln a = \int_1^a \frac{1}{x}\ dx$.
  Since $\frac{1}{x}$ is convex, the area defined by this integral is at most the area of $R$, the rectangle of width $\delta$ and height 1, minus the area of $T$, the right triangle of width $\delta$ and height $1-\frac{1}{1+\delta}$.
  Therefore
  \[
    \ln (1+\delta)
    =
    \int_1^{1+\delta} \frac{1}{x}\ dx
    <
    \delta -  \frac{1}{2} \delta \left( 1 - \frac{1}{1+\delta} \right)
    =
    \frac{\delta}{2} + \frac{\delta}{2(1+\delta)}.
    \qedhere
  \]
\end{proof}

A often-useful upper bound is $\ln (1+\delta) < \delta$, i.e., using the area of $R$ as an upper bound.
Interestingly, shaving $\delta$ down to $\frac{\delta}{2} + \frac{\delta}{2 (1+\delta)}$ by subtracting $\mathrm{area}(T)$ is crucial for proving Lemma~\ref{lem-random-walk-f-r}; using the weaker bound $\ln(1+\delta)<\delta$ in the proof of Lemma~\ref{lem-random-walk-f-r} gives only the trivial upper bound of 1 for the probability in that proof.
The same holds true for the next lemma (i.e., the term $\frac{\delta^2}{2}$ is crucial), which is employed similarly in the proof of Lemma~\ref{lem-random-walk-f-r}.


\begin{lem}\label{lem-ln-neg}
  Let $0 < \delta < 1$.
  Then
  $$
    \ln (1-\delta) < - \delta - \frac{\delta^2}{2}.
  $$
\end{lem}

\begin{proof}
  See Figure~\ref{fig:ln-neg} for an illustration of the geometric intuition.
  Recall that for $0 < a \leq 1$, $- \ln a = \int_a^1 \frac{1}{x}\ dx$.
  The area defined by this integral is at least the area of $R$, the rectangle of width $\delta$ and height $1$, plus the area of $T$, the right triangle of width $\delta$ and height $\delta$. This is because $\frac{d}{dx} \frac{1}{x} = -1$ at the value $x=1$ and $\frac{d}{dx} \frac{1}{x} < -1$ for all $0 < x < 1$, so the hypotenuse of $T$ touches the curve at $x=1$ and lies strictly underneath the curve for all $1-\delta \leq x < 1$.
  Therefore
  \[
    - \ln (1-\delta)
    =
    \int_{1-\delta}^1 \frac{1}{x}\ dx
    >
    \delta + \frac{\delta^2}{2}. \qedhere
  \]
\end{proof}

\subsection{Chernoff bound for biased random walk on $\Z$}

Let $\hf,\hr > 0$.
Let $\bU_{\hf,\hr}(t)$ be a continuous-time Markov process with state set $\Z$ in which $\bU_{\hf,\hr}(0) = 0$, with transitions from state $i$ to $i+1$ with rate $\hf$ for all $i\in\Z$, and with transitions from state $i+1$ to $i$ with rate $\hr$ for all $i\in\Z$.
In other words, $\bU_{\hf,\hr}(t)$ is a continuous time biased random walk on $\Z$, with forward bias $\hf$ and reverse bias $\hr$.

\begin{lem}\label{lem-random-walk-f-r}
  For all $\hf > \hr > 0$ and all $t,\hat{\epsilon}>0$,
  $$
    \Pr \left[ \bU_{\hf,\hr}(t) < (1-\hat{\epsilon})(\hf - \hr) t \right]
    <
    2 e^{- \frac{\hat{\epsilon}^2 (\hf - \hr)^2}{8 \hf} t }
  $$
\end{lem}

\begin{proof}
  Since the forward and reverse rates of $\bU_{\hf,\hr}$ are constants independent of the state, the total number of forward transitions $\bF$ and the total number of reverse transitions $\bR$ in the time interval $[0,t]$ are independent random variables, such that $\bF - \bR = \bU_{\hf,\hr}(t)$.
  $\bF$ is a Poisson distribution with rate $f=\hf t$ (hence $\E[\bF]=f$) and $\bR$ is a Poisson distribution with rate $r=\hr t$.

  Let $\epsilon = \hat{\epsilon} / 2$.
  Let $d = f-r > 0$.
  Let $\gamma
  = \frac{f - \epsilon d}{f}
  = \frac{(1-\epsilon)f + \epsilon r}{f}$.
  Let $\lambda = f$.
  Let $\delta = \frac{\epsilon (f-r)}{(1-\epsilon)f + \epsilon r},$ and note that $1 + \delta = \frac{f}{(1-\epsilon)f + \epsilon r}$.
  By Corollary~\ref{cor-poisson-chernoff-gamma-leq},
  \begin{align*}
    \Pr[\bF \leq \gamma \lambda]
    &\leq
    \left( \frac{e^{1-\frac{1}{\gamma}}}{\gamma} \right)^{\gamma \lambda}
    &&=
    \left( \frac{f \exp \left( 1-\frac{f}{(1-\epsilon)f + \epsilon r} \right) }{(1-\epsilon)f + \epsilon r} \right)^{(1-\epsilon)f + \epsilon r}
    \\&=
    \left( (1+\delta) \exp \left( 1- (1+\delta) \right) \right)^{(1-\epsilon)f + \epsilon r}
    &&=
    \exp \left( \ln(1+\delta) - \delta) \right)^{(1-\epsilon)f + \epsilon r}
    \\&<
    \exp \left( \frac{\delta}{2} + \frac{\delta}{2(1+\delta)} - \delta \right)^{(1-\epsilon)f + \epsilon r}
    &&\text{by Lemma~\ref{lem-ln-pos}}
    \\&=
    \exp \left( \frac{\delta - \delta(1+\delta)}{2(1+\delta)}  \right)^{(1-\epsilon)f + \epsilon r}
    &&=
    \exp \left( \frac{- \delta^2}{2(1+\delta)}  \right)^{(1-\epsilon)f + \epsilon r}
    \\&=
    \exp \left( \frac{- \left( \frac{\epsilon (f-r)}{(1-\epsilon)f + \epsilon r} \right)^2}{2 \left( 1+\frac{\epsilon (f-r)}{(1-\epsilon)f + \epsilon r} \right) }  \right)^{(1-\epsilon)f + \epsilon r}
    &&=
    \exp \left( \frac{- \left( \frac{\epsilon (f-r)}{(1-\epsilon)f + \epsilon r} \right)^2}{2 \left( \frac{(1-\epsilon)f + \epsilon r + \epsilon (f-r)}{(1-\epsilon)f + \epsilon r} \right) }  \right)^{(1-\epsilon)f + \epsilon r}
    \\&=
    \exp \left( \frac{- \left( \frac{\epsilon (f-r)}{(1-\epsilon)f + \epsilon r} \right)^2}{ \frac{2 f}{(1-\epsilon)f + \epsilon r} }  \right)^{(1-\epsilon)f + \epsilon r}
    &&=
    \exp \left( \frac{- \left( \epsilon (f-r) \right)^2}{ 2 f ((1-\epsilon)f + \epsilon r) } \right)^{(1-\epsilon)f + \epsilon r}
    \\&=
    \exp \left( \frac{- (\epsilon(f - r))^2}{2 f} \right).
    \addtag\label{ineq-F}
  \end{align*}

  Let $\gamma' = \frac{r + \epsilon d}{r} = \frac{\epsilon f + (1-\epsilon)r}{r}$ and $\lambda' = r$.
  Let $\delta = \frac{\epsilon (f - r) }{\epsilon f + (1-\epsilon)r},$ and note that $1 - \delta = \frac{r}{\epsilon f + (1-\epsilon)r}$.
  By Corollary~\ref{cor-poisson-chernoff-gamma-geq},
  \begin{align*}
    \Pr[\bR \geq \gamma' \lambda']
    &\leq
    \left( \frac{e^{1-\frac{1}{\gamma'}}}{\gamma'} \right)^{\gamma' \lambda'}
    &&=
    \left( \frac{r \exp \left( 1-\frac{r}{\epsilon f + (1-\epsilon)r} \right) }{\epsilon f + (1-\epsilon)r} \right)^{\epsilon f + (1-\epsilon)r}
    \\&=
    \left( (1-\delta) \exp \left( 1-(1-\delta) \right) \right)^{\epsilon f + (1-\epsilon)r}
    &&=
    \exp \left( \ln(1-\delta) + \delta \right)^{\epsilon f + (1-\epsilon)r}
    \\&<
    \exp \left( - \delta - \frac{\delta^2}{2} + \delta \right)^{\epsilon f + (1-\epsilon)r}
    &&\text{by Lemma~\ref{lem-ln-neg}}
    \\&=
    \exp \left( \frac{ - (\epsilon (f - r))^2 }{2 (\epsilon f + (1-\epsilon)r)^2 } \right)^{\epsilon f + (1-\epsilon)r}
    &&=
    \exp \left( \frac{ - (\epsilon (f - r))^2 }{2 (\epsilon f + (1-\epsilon)r) } \right)
    \\&<
    \exp \left( \frac{ - (\epsilon (f - r))^2 }{2 (\epsilon f + (1-\epsilon)f) } \right)
    &&=
    \exp \left( \frac{ - (\epsilon (f - r))^2 }{2 f} \right).
    \addtag\label{ineq-R}
  \end{align*}

  Observe that $\gamma \lambda - \gamma' \lambda' = f - r - 2 \epsilon (f-r) = (1 - 2 \epsilon) (f-r)$.
  By \eqref{ineq-F}, \eqref{ineq-R}, and the union bound,
  $
    \Pr[\bF - \bR \leq (1 - 2 \epsilon) (f-r)]
    <
    2 \cdot \exp \left( - \frac{\hat{\epsilon}^2 (f - r)^2}{2 f} \right).
  $
  Substituting the definitions $f=\hf t$, $r=\hr t$, and $\epsilon = \hat{\epsilon}/2$ completes the proof.
\end{proof}

\subsection{Chernoff bound for random walk on $\N$ with state-dependent reverse bias}

Let $N \in \Z^+$, let $\delta_f,\lambda_r > 0$, and let $\bW^N_{\delta_f,\lambda_r}(t)$ for $t > 0$ be a continuous-time Markov process with state set $\N$ in which $\bW^N_{\delta_f,\lambda_r}(0) = 0$, with transitions from state $i$ to state $i+1$ with rate $\delta_f N$ for all $i\in\N$, and with transitions from state $i+1$ to $i$ with rate $\lambda_r i$ for all $i\in\N$.
In other words, $\bW^N_{\delta_f,\lambda_r}$ is a continuous time random walk on $\N$ with a reflecting barrier at 0, in which the rate of going forward is a constant $\delta_f N$, and the rate of going in reverse from state $i$ is proportional to $i$ (with constant of proportionality $\lambda_r$).


\begin{lem}\label{lem-random-walk-reflecting}
  Let $\lambda_r \geq 1$, let $\delta_f,\delta_r > 0$ such that $\delta_r \leq \frac{\delta_f}{4 \lambda_r}$.
  Then for all $N\in\Z^+$ such that $N \geq 6 / (\delta_f)$,
  $$
    \Pr \left[ \max_{\hatt \in [0,1]} \bW^N_{\delta_f,\lambda_r}(\hatt) < \delta_r N \right]
    <
    2^{-\delta_f N / 22 + 1}.
  $$
\end{lem}

\begin{proof}
  Consider the event that $(\forall \hatt \in [0,1])\ \bW^N_{\delta_f,\lambda_r}(\hatt) < \delta_r N$.
  Then in this case, the maximum reverse transition rate is at most $\lambda_r \delta_r N \leq \delta_f N / 4$.

  For $\hatt \in [0,1]$, consider the random walk $\bU_{\delta_f N,\delta_f N / 4}(\hatt)$ defined as in Lemma~\ref{lem-random-walk-f-r} as a Markov process on $\Z$ (i.e., the states are allowed to go negative) in which the forward rate from any state $i\in\Z$ is $\delta_f N$ as in $\bW^N_{\delta_f,\lambda_r}(\hatt)$, but the reverse rate from any state $i\in\Z$ is $\delta_f N / 4$, which is an upper bound on the reverse rate of $\bW^N_{\delta_f,\lambda_r}(\hatt)$ from any state $i \in \{1,\ldots,\delta_r N \}$.

  Therefore $\Pr[(\forall \hatt \in [0,1])\ \bW^N_{\delta_f,\lambda_r}(\hatt) < \delta_r N] < \Pr[(\forall \hatt \in [0,1])\ \bU_{\delta_f N,\delta_f N / 4}(\hatt) < \delta_r N],$ since $\bU_{\delta_f N,\delta_f N / 4}$ has a strictly higher reverse rate and does not have a reflecting barrier at state $i=0$, hence is strictly less likely never to reach the state $\delta_r N$ at any time $\hatt \in [0,1]$.

  We prove the theorem by bounding $\Pr[\bU_{\delta_f,\delta_f N / 4}(1) < \delta_r N]$.

  Lemma~\ref{lem-random-walk-f-r}, with $\hat{\epsilon} = \frac{2}{3}$, $t=1$, $\hf = \delta_f N$ and $\hr = \delta_f N / 4$ implies that
  \begin{align*}
    \Pr \left[ \bU_{\delta_f N,\delta_f N / 4}(1) < \frac{\delta_f N}{4} \right]
    &=
    \Pr \left[ \bU_{\delta_f N,\delta_f N / 4}(1) < \left(1 - \frac{2}{3} \right) (\delta_f N - \delta_f N / 4) \right]
    \\&<
    2 e^{- \frac{\left(\frac{2}{3}\right)^2 (\delta_f N - \delta_f N / 4)^2}{8 \delta_f N} } \ \ \ \ \text{by Lemma~\ref{lem-random-walk-f-r}}
    \\&=
    2 e^{- \frac{(3 \delta_f / 4)^2}{18 \delta_f} N}
    =
    2 e^{- \delta_f N / 32}.
  \end{align*}
  Note that $e^{-n/32} = 2^{-n/(32 \ln 2)} < 2^{-n/22}$ for all $n > 0$.
  Therefore,
  $$
    \Pr \left[ \bU_{\delta_f N,\delta_f N / 4}(1) < \frac{\delta_f N}{4} \right]
    <
    2 \cdot 2^{-\delta_f N / 22}
    =
    2^{-\delta_f N / 22 + 1}.
  $$
  Since $\delta_f \geq 4 \lambda_r \delta_r$ and $\lambda_r \geq 1$,
  $
    \delta_f N / 4
    \geq
    \lambda_r \delta_r N
    \geq
    \delta_r N,
  $
  so
  $$
    \Pr \left[ \bU_{\delta_f N,\delta_f N / 4}(1) < \delta_r N \right]
    \leq
    2^{-\delta_f N / 22 + 1}.
  $$
  By our observation that $\Pr[(\forall \hatt \in [0,1])\ \bW^N_{\delta_f,\lambda_r}(\hatt) < \delta_r N] < \Pr[(\forall \hatt \in [0,1])\ \bU_{\delta_f N,\delta_f N / 4}(\hatt) < \delta_r N],$ and since $\bU_{\delta_f N,\delta_f N / 4}(1) \geq \delta_r N$ is a counterexample to the latter event with $\hatt=1$,
  $$
    \Pr \left[ \max_{\hatt \in [0,1]} \bW^N_{\delta_f,\lambda_r}(\hatt) < \delta_r N \right]
    \leq
    2^{-\delta_f N / 22 + 1},
  $$
  completing the proof.
\end{proof}

\end{document}